\documentclass[11pt,reqno]{amsart}
\usepackage{amsfonts}
\usepackage{amssymb}
\usepackage{textcomp}
\usepackage{graphicx}
\def\sqr#1#2{{\vcenter{\vbox{\hrule height.#2pt
        \hbox{\vrule width.#2pt height#1pt \kern#2pt
        \vrule width.#2pt}
        \hrule height.#2pt}}}}

\newcommand{\nc}{\newcommand}
\nc{\parent}[1]{$[\![#1]\!]$}
\newtheorem{theorem}{Theorem}[section]
\newtheorem{lemma}{Lemma}[section]

\newtheorem{corollary}{Corollary}[section]
\newtheorem{proposition}{Proposition}[section]
\newtheorem{remark}{Remark}
\newtheorem{definition}{Definition}[section]
\newtheorem{assumption}{Assumption}[section]

\newenvironment{pf-main}{{\sc Proof of Theorem \ref{mainresult}.}\hspace{3mm}}{\qed}
\nc{\eid}{\stackrel{d}{=}}
\nc{\cadlag}{c\`{a}dl\`{a}g } \nc{\ba}{\begin{array}}
\nc{\ea}{\end{array}} \nc{\be}{\begin{equation}}
\nc{\ee}{\end{equation}} \nc{\bea}{\begin{eqnarray}}
\nc{\eea}{\end{eqnarray}} \nc{\bean}{\begin{eqnarray*}}
\nc{\eean}{\end{eqnarray*}} \nc{\bu}{\bullet} \nc{\nn}{\nonumber}
\nc{\cA}{{\mathcal A}} \nc{\cB}{{\mathcal B}} \nc{\cC}{{\mathcal
C}} \nc{\cD}{{\mathcal D}} \nc{\bbD}{\mathbb{D}}
\nc{\cG}{{\mathcal G}} \nc{\cF}{{\mathcal F}} \nc{\cS}{{\mathcal
S}} \nc{\cU}{{\mathcal U}} \nc{\cH}{{\mathcal H}}
\nc{\cK}{{\mathcal K}}\nc{\cL}{{\mathcal L}}  \nc{\cM}{{\mathcal
M}} \nc{\cO}{{\mathcal O}} \nc{\cP}{{\mathcal P}}
\nc{\bbE}{\mathbb{E}} \nc{\bbF}{\mathbb{F}}
\nc{\bbEQ}{\mathbb{E}_{\mathbb{Q}}} \nc{\eps}{\varepsilon}
\nc{\bbEP}{\mathbb{E}_{\mathbb{P}}}\nc{\bbL}{\mathbb{L}}
\nc{\what}{\widehat} \nc{\bbP}{\mathbb{P}} \nc{\bbQ}{\mathbb{Q}}
\nc{\del}{\partial} \nc{\Om}{\Omega} \nc{\om}{\omega}
\nc{\bbR}{\mathbb{R}} \nc{\bbN}{\mathbb{N}} \nc{\fps}{$(\Om, \cF,
(\cF_t)_{t\geq 0}, \bbP)$} \nc{\bbC}{\mathbb{C}}
\nc{\bfr}{\begin{flushright}} \nc{\efr}{\end{flushright}}
\nc{\dXt}{\Delta X_{t}} \nc{\dXs}{\Delta X_{s}}
\nc{\bs}{\blacksquare} \nc{\dX}{\Delta X} \nc{\dY}{\Delta Y}
\nc{\dnkx}{\left(X(T^{n}_{k})-X(T^{n}_{k-1})\right)}
\nc{\esssup}{\mathrm{ess}\mbox{ }\mathrm{sup}}
\nc{\essinf}{\mathrm{ess}\mbox{ } \mathrm{inf}}
\nc{\dhats}{\widehat{\delta_s}} \nc{\half} {\frac{1}{2}}

\nc{\ol}{\overline}
\def\rar{\rightarrow}

\nc{\chf}{\mbox{$\mathbf1$}}
\begin{document}

\title{Financial equilibrium with asymmetric information and random horizon}
\author{Umut \c{C}etin}
\address{Department of Statistics, London School of Economics and Political Science, 10 Houghton st, London, WC2A 2AE, UK}
\email{u.cetin@lse.ac.uk}
\date{\today}
\begin{abstract}
We study in detail and explicitly solve the version of Kyle's model introduced in a specific case in \cite{BB}, where the trading horizon is given by an exponentially distributed random time. The first part of the paper is devoted to the analysis of time-homogeneous equilibria using  tools from the theory of one-dimensional diffusions. It turns out that such an equilibrium is only possible if the final payoff is Bernoulli distributed as in \cite{BB}. We show in the second part that the signal of the market makers use in the general case is a time-changed version of the one that they would have used had the final payoff had a Bernoulli distribution. In both cases we characterise explicitly the equilibrium price process and the optimal strategy of the informed trader. Contrary to the original Kyle model it is found that the reciprocal of market's depth, i.e. Kyle's lambda, is a uniformly integrable supermartingale. While Kyle's lambda is a potential, i.e. converges to $0$,  for the Bernoulli distributed final payoff, its limit in general is different than $0$. 
\end{abstract}
\maketitle

\section{Introduction}
The canonical model of markets with asymmetric information is due to
Kyle \cite{K}. Kyle studies a market for a single risky asset whose
price is determined in equilibrium. There are mainly three types of
agents that constitute the market: a strategic risk-neutral informed
trader with a private information regarding the future value of the
asset, non-strategic uninformed noise traders, and a number of risk-neutral
market makers competing for the net demand from the strategic and
non-strategic traders. The key feature of this model is that the asset value becomes public knowledge at a fixed deterministic date and the market makers cannot distinguish  the
informed trades from the uninformed ones but `learn'
from the net demand by `filtering' what the informed trader knows,
which is `corrupted' by the noise demand. The private information of
the informed trader is static, i.e. it is obtained at the beginning of
the trading and does not change over time. The nature of this
information is not really important: It could be inside
information about the future payoff of the asset or an
unbiased estimator of the future payoff. The latter is a more suitable
interpretation when the strategic informed trader is a big
investment bank with a good research division producing sophisticated
research that is not shared with the public.

Kyle's model is in discrete time and assumes that the noise traders
follow a random walk and the future payoff of the asset has a normal
distribution. This has been extended to a continuous time framework with general payoffs by
Back \cite{B}. In this extension the total demand of the noise traders is given by a
Brownian motion and the future payoff of the asset has a general
continuous distribution while the informed trader's private
information is still static.  Kyle's model and its continuous-time
extension by Back have been further extended to allow multiple
informed traders (\cite{fw}, \cite{bcw}), to include default risk \cite{CampiCetin} or to the case when the  single informed trader
receives a continuous signal as private information (\cite{BP} and \cite{GBP}).

Our first goal in this paper is to study a time-homogeneous  version of this model introduced by Back and Baruch in \cite{BB}. The time-homogeneity refers to the SDE corresponding to the market price having time homogeneous coefficients, and the insider's optimal decisions depending only on the price of the traded asset - not on time.  This is in part due to the assumption that the asset value, $\Gamma$, is announced at a random exponential time, $\tau$, with mean $r^{-1}$ for some $r>0$ and independent from all other variables in the model. The informed trader knows  $\Gamma$ but not $\tau$. Thus, she has an informational advantage over the other traders; however, this can end at any time since $\tau$ will come as a suprise to the informed trader as it does to the others.  Back and Baruch compute the market depth and the informed trader's strategy as a function of price, which is only characterised implicitly using an inverse operation. Moreover, their method does not allow an explicit  characterisation of the distribution of the equilibrium price process. 

In Section \ref{s:Bernoulli} we analyse the model of Back and Baruch using tools from the theory of one-dimensional diffusions. We show that in equilibrium the market makers construct  a transient Ornstein-Uhlenbeck process, $Y$, that they will use for pricing. A particular consequence of this is that  the pricing rule becomes a {\em scale function} of $Y$. In addition we  identify the value function, $J$, of the strategic trader with an {\em $r$-excessive function} of a certain {\em $h$-transform} of the Ornstein-Uhlenbeck process. 

Back and Baruch postulate the equilibrium controls for the market makers and the informed trader and then verify that these indeed constitute an equilibrium. We, on the other hand, study in detail the admissible pricing rules for the market makers given a large class of admissible strategies of the informed trader and characterise the ones that can appear in the equilibrium. It turns  out that the market makers can choose from a continuum of controls in the equilibrium; however, every such choice will lead to the same SDE for the equilibrium price process.

Consistent with what was observed earlier by Back and Baruch we show that the process measuring the price impact of trades, i.e. {\em Kyle's lambda}, is a uniformly integrable supermartingale converging to $0$, i.e. a potential.  As a result, the market gets more liquid on average as time passes. This is a deviation from  Kyle \cite{K} who predicted that the Kyle's lambda must follow a martingale preventing `systematic changes' in the market depth as explained in \cite{bcw} and \cite{CD}.  In the absence of such systematic changes, the insider cannot acquire a large position when the depth is low to liquidate at a later date when the liquidity is higher to obtain unbounded profits. In the model we study, even if the market gets more liquid as time passes,  there are  no opportunities for the informed trader to make infinite profits since the market can end at any time interval $[t,t+dt]$ with probability $rdt$.  

One of the advantages of the approach used in this paper is that it identifies the distribution of the price process explicitly.  The explicit form of the solutions,  however, also indicate that one cannot have an equilibrium, where the price process is a time-homogeneous  diffusion, if the asset value has a non-Bernoulli distribution as explained in Remark \ref{r:extension}.  A recent work by Collin-Dufresne et al. \cite{CFM}, on the other hand, hints at the direction that one should follow for a general payoff distribution.  Collin-Dufresne et al. \cite{CFM} study a version of the Kyle model where the announcement date is a jump time of a Poisson process with non-constant intensity and the asset value has a Gaussian distribution using linear Kalman filtering. It turns out in their model that  the diffusion coefficient of the SDE for the equilibrium price should be time dependent. 

Motivated by their result we study in Section \ref{s:general} an extension of the model of Back and Baruch to a large class of payoff distributions. It turns out that the signal that the market makers use when the payoff has a general distribution is a time-changed transient Ornstein-Uhlenbeck process. The time change, $V(t)$, is deterministic with $V(\infty)<\infty$. The finiteness of the time-change implies that the limiting distribution of the market makers' signal is a non-degenerate normal distribution. This particular feature allows us to extend the model of Back and Baruch to a much more general setting, including the normally distributed case considered in \cite{CFM}.  

As in the other works on the Kyle model we  establish that the informed trader's trades are {\em inconspicuous} in the equilibrium. That is, the equilibrium distribution of the total demand is the same as that of cumulative noise trades. An essential difference, on the other hand, from the earlier works on this subject is that the equilibrium prices exhibit a jump at the announcement date, $\tau$. This is only natural since  $\tau$ is unknown to the informed trader and, thus, there is no strategy to  ensure that the market price converges almost surely to $\Gamma$ as time approaches to  $\tau$, which is a totally inaccessible stopping time even for the informed. However, what the informed trader can do is the following strategy, which will in fact turn out to be her equilibrium strategy: She can make sure that $P_t$, the market price conditioned on {\em survival}, i.e., $[\tau>t]$, converge to $\Gamma$ as $t \rar \infty$. That is, conditioned on indefinite survival market prices converge to the true payoff. Note that indefinite survival has $0$ probability due to the finiteness of $\tau$, hence a jump in prices at time $\tau$ occur with probability $1$.

The outline of the paper is as follows: Section \ref{s:setup} introduces the model, defines the admissible controls for the market makers as well as the informed trader, and ends with a characterisation of the market makers' pricing choice.  The equilibrium when the final payoff has a  Bernoulli and a general distribution are solved in Sections \ref{s:Bernoulli} and \ref{s:general}, respectively. Finally, Section \ref{s:conc} concludes.

\section{The setup} \label{s:setup}
Let $(\Omega , \cF , (\cF_t)_{t \geq 0} , \bbP)$ be a filtered probability space  satisfying the usual conditions of right continuity and $\bbP$-completeness. We suppose that $\cF_0$ is not trivial and there exists an $\cF_0$-measurable random variable, $\Gamma$, taking values in $\bbR$. Moreover, the filtered probability space also supports a standard Brownian motion, $B$, with $B_0=0$ and, thus, $B$ is independent of $\Gamma$. We also assume the existence of an $\cF$-measurable random variable, $\tau$, which is independent of $\cF_{\infty}$ and has exponential distribution with mean $0<r^{-1}<\infty$. In particular, $\tau$ is independent of $\Gamma$ as well as $B$.

The  measure induced by $\Gamma$ on $\bbR$ will be denoted by $\nu$, i.e. $\nu(A):=\bbP(\Gamma \in A)$ for any Borel subset of $\bbR$.   We further assume the existence of a family of probability measures, $(\bbP^v)_{v \in \bbR}$, such that the following disintegration formula holds:
\be \label{e:disintegration}
\bbP(E)=\int_{\bbR}\bbP^v(E)\nu(dv), \qquad \forall E \in \cF.
\ee
The existence of such a family is easily justified when we consider $\Om =\bbR \times C(\bbR_+,\bbR)$, where $C(\bbR_+,\bbR)$ is the space of real valued continuous functions on $\bbR_+$. We set $\bbP^v =\bbP$ if $v \notin \mbox{supp}(\nu)$. We also assume that $\Gamma$ is square integrable. That is,
\be \label{a:GamInt}
 \bbE[\Gamma]=\int_\bbR v^2\nu(dv)<\infty.
\ee

We consider a market in which the risk free interest rate is set to $0$ and a single risky asset is traded. The fundamental value of this asset equals $\Gamma$, which  will be  announced at the random time $\tau$. 

There are three types of agents that interact in this market:
\begin{itemize}
\item[i)] Liquidity traders who trade for reasons exogenous to the model and whose cumulative demand at time $t$ is given by  $B_t$.
\item[ii)] A single informed trader,  who knows $\Gamma$ from time $t=0$ onwards, and is risk neutral. We will call the informed trader {\em insider}  in what follows and denote her cumulative demand at time $t$ by $\theta_t$. The filtration of the insider, $\cG^I$, is generated by observing the price of the risky asset, $\Gamma$, and whether the announcement has been made, that is, whether $\tau>t$ or not for each $t\geq 0$. The filtration will also be assumed to be completed with the null sets of $(\bbP^v)_{v \in \bbR}$. 
\item[iii)] Market makers observe only the net demand of the risky asset, $X=B+\theta$,  whether $\tau>t$ or not for each $t\geq 0$, and $\Gamma$ when it is made public at $\tau$. Thus, their filtration, $\cG^M$, is the minimal right-continuous filtration generated by $X$, $(\chf_{[\tau>t]})_{t \geq 0}$, and $(\chf_{[t \geq \tau]}\Gamma)_{t \geq 0}$, and  completed with the  $\bbP$-null sets. In particular, $\tau$ is a $\cG^M$-stopping time. The price process chosen by the market makers is denoted by $S$. Obviously, $S_t =\Gamma$ on the set $[t \geq \tau]$. Prior to the announcement date we assume that the market price is determined according to a {\em Bertrand  competition:} Market makers make their price offers and the investors trade with the one offering the best quote. This mechanism results in $S$ being a $\cG^M$-martingale, i.e. $S_t=\bbE[\Gamma|\cG^M_s]$. Consequently, $S_{\infty}$ exists and equals $\Gamma$ since $\bbP(\tau<\infty)=1$.  
\end{itemize}

The way that the price process is determined yields the following since $\Gamma$ is integrable.
\begin{proposition} \label{p:SUI} In view of the condition (\ref{a:GamInt}), for any $s \geq 0$
\be \label{e:S-Glimit}
\lim_{t \rar \infty}\bbE[|S_t-\Gamma||\cG^M_s]=0.
\ee
\end{proposition}
\begin{proof}
Since $\lim_{t \rar \infty}S_t =\Gamma$ and $\Gamma$ satisfies (\ref{a:GamInt}), we deduce that the family $(S_t)_{t \geq 0}$ is uniformly integrable (see Problem 1.3.20 in \cite{KS}) and so is $(\Gamma -S_t)_{t \geq 0}$. Since $\lim_{t \rar \infty}|S_t-\Gamma|=0$, the claim follows. 
\end{proof}

In this paper, as in all other past work on Kyle's model, we are interested in Markovian Nash equilibria. In line with the current literature, we shall assume that the price chosen by the market makers is a deterministic function of some process, $Y$, which solves $dY_t=w(t,Y_t)dX_t +b(t,Y_t)dt$ for some weighting function, $w$, and drift function, $b$, chosen by the market makers (see, e.g., \cite{B, BP, GBP} among others for the use of a weighting function in the construction of market makers' signal). 

Before defining what we mean precisely by an equilibrium in this model we introduce the class of admissible  controls for the  market makers and the informed trader. 
\begin{definition} The pair $\left((w, b, y),  h\right)$ is an admissible pricing rule for the market makers if for some interval $(l,u)\subset \bbR$, $y \in (l,u)$, and $w:\bbR_+\times (l,u)\mapsto (0,\infty)$, $b:\bbR_+\times (l,u)\mapsto \bbR$, and  $h:\bbR_+\times (l,u)\mapsto \bbR$ are measurable functions such that $h(t,\cdot)$ is strictly increasing for every $t >0$, $w$ is bounded away from $0$ on compact subsets of $\bbR_+\times(l,u)$, and  for any Brownian motion, $\beta$,  there exists a unique strong solution without explosion\footnote{That is, the boundary points $l$ and $u$ are not to be reached in finite time.} to 
\[
Y_t= z+ \int_0^t w(s,Y_s) d\beta _s + \int_0^t b(s,Y_s) ds, \qquad  \forall z \in (l,u).
\]
\end{definition}

Given an admissible pricing rule, $\left((w,b, y), h\right)$, the market makers will set the price on $[\tau>t]$ to be $h(t,Y_t)$, where $Y$ follows
\be \label{e:Yadmissible}
Y_t= y+ \int_0^t w(s,Y_s)\left\{dB_s +d\theta_s\right\} + \int_0^t b(s,Y_s)ds,
\ee
whenever $\theta$ is an admissible strategy for the insider.

Since $h$ is strictly monotone, $Y$ is perfectly observable by the traders, in particular by the insider.  In conjunction with the assumption that $w$ is bounded away from $0$ on compact subsets of $\bbR_+\times (l,u)$ this will entail that  the insider can observe $B$ since she clearly knows  her own strategy, $\theta$. This assumption, thus, also ensures that the insider's filtration is well-defined and is generated by $B, \Gamma, $ and $(\tau \wedge t)_{t\geq 0}$. This is worth noting since, otherwise, we may run into problems with the well-posedness of the model as the insider's trading strategy may depend on the observations of the price process that  is adapted to the filtration generated by $X$, which depends crucially on the insider's trading strategy. 
\begin{definition} Given an admissible pricing rule  $\left((w,b, y), h\right)$,    $\theta$ is an admissible strategy for the insider if $\theta$ is of finite variation and the following conditions are satisfied.
\begin{enumerate} 
\item $\theta$ is absolutely continuous:
\[
\theta_t=\int_0^t \alpha_s ds,
\]
for some  $\cG^I$-adapted $\alpha$, where $\cG^I$ is the minimal right continuous filtration generated by $B,\Gamma$ and  $(\tau \wedge t)_{t\geq 0}$, and completed with the  $(\bbP^v)_{v \in \bbR}$-null sets.
\item There exists a unique strong solution to 
\[
Y_t=  z+ \int_0^t w(s,Y_s)dX_s + \int_0^t b(s,Y_s)ds,\qquad \forall z \in (l,u).
\]
\item The following integrability condition holds:
\be \label{adm:alpha}
\bbE^v\int_0^{\tau} h^2(t,Y_t)dt<\infty, \qquad \forall v \in \mbox{supp}(\nu),
\ee
where $Y$ is the unique strong solution of (\ref{e:Yadmissible}).
\end{enumerate}
\end{definition}

The assumption that the strategy is absolutely continuous is without any loss of generality since any strategy with positive quadratic variation is necessarily suboptimal due to the price impact of the trades (see \cite{B} for a proof of this fact). 

Faced with an admissible pricing rule  $\left((w,b, y),  h\right)$ the insider employs an admissible strategy, $\theta$, and achieves 
\[
W_{\tau}=\int_{[0,\tau]} \theta_s dh(s,Y_s) + \theta_{\tau} (\Gamma-h(\tau,Y_{\tau})),
\]
as its total profit when the public announcement is made and the trading possibilities end, where $Y$ is the strong solution to (\ref{e:Yadmissible}), which exists and is unique since $\theta$ is admissible. Note that the term $\theta_{\tau} (\Gamma-h(\tau,Y_{\tau}))$ is due to a potential jump in the price when the true value is revealed. Integrating by parts we obtain a more convenient representation:
\[
W_{\tau}=\int_0^{\tau}(\Gamma - h(s,Y_s))\alpha_s ds.
\]
Being risk-neutral the informed trader's goal is to maximise $\bbE^v W_{\tau}$ within the class of admissible strategies. 
\begin{definition} $\left((w^*,b^*,y^*),h^*,\alpha^*\right)$ is an equilibrium if 
\begin{enumerate}
\item $\left((w^*,b^*,y^*),h^*\right)$ is an admissible pricing rule for the market makers;
\item $\theta^*$ defined by $\theta_t^*=\int_0^t \alpha_s^* ds$ is an admissible trading strategy for the insider given $\left((w^*,b^*,y^*),h^*\right)$;
\item $S_t=\chf_{[\tau>t]}h^*(t,Y^*_t)+ \chf_{[\tau\leq t]}\Gamma$ is a $\cG^M$-martingale, where $Y^*$ is the unique strong solution of (\ref{e:Yadmissible}) with $w=w^*$ and $b=b^*$;
\item $\theta^*$ maximises the expected profits for the insider. That is,
\[
\bbE^v \int_0^{\tau}(v - h^*(s,Y^*_s))\alpha_s^* ds = \sup_{\alpha \in \cA}\bbE^v\int_0^{\tau}(v - h^*(s,Y^*_s))\alpha_s ds, \qquad \forall v \in \mbox{supp}(\nu),
\] 
where $\cA$ is the class of all admissible strategies given the admissible pricing rule $\left((w^*,b^*,y^*),h^*\right)$.
\end{enumerate}
\end{definition}
Since $\tau$ is independent of $\Gamma$ and $B$, the market makers as well as the insider will not use controls that depends on $\tau$ in the equilibrium. In this case the expected value of the  final wealth of the insider will equal
\be \label{e:evalins}
\bbE^v[W_{\tau}]=\int_0^{\infty}e^{-rs}\bbE^v\left[(v - h(s,Y_s))\alpha_s \right]ds.
\ee

 Let us denote by  $\cF^X$  the minimal right continuous filtration containing the $\bbP$-null sets and with respect to which $X$ is adapted.  When the insider uses strategies that are independent of $\tau$, this will render $X$  independent of $\tau$ as well. In this case one should expect that
 \[
 S_t= P_t \chf_{[t<\tau]} + \Gamma \chf_{[t \geq \tau]},
 \]
 where $P$ is a semimartingale adapted to $\cF^X$. The following proposition shows that this is indeed the case and gives a characterisation of $P$ in terms of $\Gamma$. 
 \begin{proposition} \label{p:PV} Suppose that $X$ is independent of $\tau$ and all $\cF^X$-martingales are continuous. Then, $P_t=\bbE[\Gamma|\cF^X_t]$, i.e. $P$ is the $\cF^X$-optional projection of $\Gamma$. In particular, $P$ is continuous.
 \end{proposition}

The proof of the above  result is delegated to  Appendix \ref{a:p}. However, it is worth to mention here that the proof does not rely on market makers making their pricing decision in a Markovian manner, i.e. $P_t=h(t,Y_t)$, where $Y$ follows (\ref{e:Yadmissible}). That is, whatever the pricing decision is, it must satisfy $P_t=\bbE[\Gamma|\cF^X_t]$ whenever $X$ is independent of $\tau$ and all $\cF^X$-martingale are continuous\footnote{As we shall see in the next section all $\cF^X$-martingales are continuous in the equilibrium. In general a sufficient conditions for this property is $\bbE\int_0^t\alpha_s^2ds<\infty$ for all $t>0$ (see Corollary 8.10 of \cite{RW1}).}.
 
In view of the above characterisation of the pricing decision, $P$, of the market makers, we will search for an equilibrium in the next section by studying the optimal trading choices of the insider. As our focus is not on the uniqueness of the equilibrium, we shall follow our intuition and consider trading strategies that are independent of $\tau$.  Accordingly, the following section will establish the existence of an equilibrium, where $X$ is independent of $\tau$. 
\section{Insider's optimisation problem and the equilibrium}
\subsection{Bernoulli distributed liquidation value} \label{s:Bernoulli} In this section we assume that $\Gamma$ takes values in $\{0,1\}$ as in \cite{BB} and set $p:=\nu(\{1\})$, i.e the probability that the liquidation value  equals $1$.

Under this assumption we shall next see that one can obtain an equilibrium where the coefficients in (\ref{e:Yadmissible}) as well as the pricing function, $h$, do not depend on time. Consequently,   the solution of the optimisation problem for the insider will be time-homogeneous as well. 

We first try to formally obtain the Bellman equations associated to the value function of the insider. To this end, as observed above, suppose that $X$ is independent of $\tau$  and the market makers set the price at time $t$ on $[t<\tau]$ to be $h(Y_t)$ for some sufficiently smooth $h$, where
\[
dY_t= a(Y_t)dX_t + \phi(Y_t)dt.
\]
Also recall that $dX_t=dB_t+\alpha_t dt$, where $\alpha_t dt$ represents the infinitesimal trades of the insider. 

In view of (\ref{e:evalins}) we may consider the following value function for the insider's problem given that $\Gamma=v$:
\be \label{e:Jber}
J(x)=\sup_{\alpha}\bbE^v\left[\int_0^{\infty}e^{-rs}(v-h(Y_s))\alpha_s ds\bigg|Y_0=x\right].
\ee
Formally, $J$ solves
\[
\sup_{\alpha}\left\{\frac{1}{2}a^2 J'' + J' (a \alpha +\phi)+  \alpha (v-h) -rJ\right\}=0.
\]
Thus, if it is three times continuously differentiable, we expect to have
\bea
a J' &=& h-v \label{e:J'}\\
\frac{1}{2}a^2 J'' + J' \phi -rJ&=&0. \label{e:AJ}
\eea
The first identity yields
\[
J'=\frac{h-v}{a}, \; J''=\frac{h'}{a}- \frac{h-v}{a^2}a ' , \mbox{ and } J'''=\frac{h''}{a}-2\frac{h' a'}{a^2}-\frac{(h-v)a''}{a^2}+2\frac{(h-v)(a')^2}{a^3}.
\]
Plugging the above into the second identity yields
\[
0= \frac{a^2}{2}h''+h' \phi -(h-v)\left[\frac{a a''}{2}+\frac{a'}{a}\phi-\phi' + r\right].
\]
Since $h$ cannot depend on $\Gamma$, we obtain the following conditions for the candidate pricing rule $h$ and the coefficients $a$ and $\phi$ that will be chosen by the market makers to construct the price.
\bea
\frac{1}{2}a^2 h'' + h' \phi&=& 0 \label{e:hrule} \\
\frac{1}{2}a^2a'' +a'\phi +(r-\phi')a&=& 0 \label{e:phia}.
\eea

Looking at this equation we may guess that in the equilibrium 
\be \label{e:sdeY}
dY_t= a(Y_t)dB^Y_t + \phi(Y_t)dt,
\ee
where $B^Y$ is a $\cG^M$-Brownian motion, and  $a$ and $\phi$ solve (\ref{e:phia}). The following shows that all such processes can be obtained from an Ornstein-Uhlenbeck process.
\begin{proposition} \label{p:YOU} Suppose that $Y$ is a regular one-dimensional diffusion on $(l,u)$ defined by the generator
\[
\frac{1}{2}a^2 \frac{d^2}{dx^2}+\phi \frac{d}{dx},
\]
where $a> 0$ and $\phi$ are two functions satisfying (\ref{e:phia}) on $(l,u)$. Assume further that for any $y \in (l,u)$ there exists a unique weak  solution of (\ref{e:sdeY}) and the following integrability condition holds for some $c \in (l,u)$
\[
f(x):=\int_c^x \frac{1}{a(y)} <\infty, \, \forall x \in (l,u).
\]
If $R_t= f(Y_t)$, then  $R$ is an Ornstein-Uhlenbeck process with the generator
\be \label{e:OUgen}
A=\frac{1}{2} \frac{d^2}{dx^2}+ (rx+d) \frac{d}{dx},
\ee
for some $d \in \bbR$.
\end{proposition}
\begin{proof}
Clearly, $Y$ solves (\ref{e:sdeY}) with some Brownian motion $B^Y$. Thus, it follows from Ito's formula that
\[
dR_t= dB^Y_t + \left\{\frac{\phi(f^{-1}(R_t))}{\sigma(f^{-1}(R_t))}-\frac{1}{2}\sigma'(f^{-1}(R_t))\right\}dt.
\]
Using (\ref{e:phia}) one can easily check that the derivative of the function
\[
y \mapsto \frac{\phi(f^{-1}(y))}{\sigma(f^{-1}(y))}-\frac{1}{2}\sigma'(f^{-1}(y))
\]
equals $r$, which yields the claim.
\end{proof}

In view of the above proposition we expect the equilibrium price process to be a function of an Ornstein-Uhlenbeck process. Thus, there is no harm in choosing $a\equiv 1$ and $\phi=rx+d$ for some $d \in \bbR$. This means that the pricing rule, $h$, will be a scale function of the diffusion with the generator (\ref{e:OUgen}). Recall that the choice of $a$ and $\phi$ are made by the market makers. Thus, the market makers  will construct their signal, $Y$, such that it is a transient process in the equilibrium ($|Y_t|\rar \infty$ when $Y$ is defined by (\ref{e:OUgen})).

\begin{lemma}
Define
\be \label{e:OUscale}
s(x)=\sqrt{\frac{r}{\pi}}\int_{-\infty}^x \exp\left(-r\left(y+\frac{d}{r}\right)^2\right)dy.
\ee
Then $s$ is a scale function for the diffusion defined by (\ref{e:OUgen}) with $s(-\infty)=1-s(\infty)=0$.
\end{lemma}
\begin{proof}
It is straightforward to check that $A s=0$. Moreover, $s$ is the cumulative distribution function for a normal random variable implying the boundary conditions.
\end{proof}

\begin{theorem} \label{t:optY} Let $E\in \cF_0$ and consider the SDE
\[
dY_t= dB_t +\left\{rY_t +d +\chf_E \frac{s'(Y_t)}{s(Y_t)} -\chf_{E^c}\frac{s'(Y_t)}{1-s(Y_t)}\right\}dt, \; Y_0=y \in \bbR,
\]
where $s$ is as given by (\ref{e:OUscale}).  Then, there exists a unique strong solution to the above. The solution, in particular,  satisfies 
\[
\lim_{t \rar \infty} Y_t(\om)=\left\{\ba{rl} \infty, &\mbox{if } \om \in E; \\
-\infty, & \mbox{if } \om \in E^c. \ea\right.
\]
Moreover, if $\bbP(E)=s(y)$, we have
\[
\bbP(E|\cF^Y_t)=s(Y_t),
\]
and, consequently,
\[
dY_t= dB^Y_t +\left\{rY_t +d \right\}dt
\]
in its own filtration, where $B^Y$ is a Brownian motion.
\end{theorem}
\begin{proof}
Consider an  Ornstein-Uhlenbeck process, $R$, with the generator (\ref{e:OUgen}) in some probability space $(\Om, \cG, (\cG_t)_{t \geq 0}, P)$ such that there exists a set $ F \in \cG_0$ with $P(F)=\bbP(E)$. Clearly, $M:=\chf_F \frac{s(R)}{s(y)} + \chf_{F^c}\frac{1-s(R)}{1-s(y)}$ is a bounded martingale with $E[M_0]=1$. Thus, defining $Q$ on $\cG$ by $\frac{dQ}{dP}=M_{\infty}$ yields the existence of a weak solution. Moreover, the weak solution is unique in law since $M_t>0$ for $t>0$. The limiting condition for $Y$ as $t \rar \infty$  follows from this construction of the weak solution since the construction is nothing but the $h$-transform that achieves $R_{\infty}=\infty$ (resp. $R_{\infty}=-\infty$) on $F$ (resp.  $F^c$) (see Paragraph 32 on p.34 of \cite{BorSal}). 

The SDE above in  fact  possesses a unique strong solution. Indeed,  since $\sigma \equiv 1$,  Lemma IX.3.3,  Corollary IX.3.4 and Proposition IX.3.2 in \cite{RY} imply that pathwise uniqueness holds for this SDE.  Thus, in view of the celebrated result of Yamada and Watanabe (see Corollary 5.3.23 in \cite{KS}), there exists a unique strong solution.

Moreover,  using the law of the solutions obtained via the weak construction performed above, we have
\[
\bbP[E|\cF^Y_t]=\frac{E^y\left[M_{\infty}\chf_{F}|\cF^R_t\right]}{E^y[M_{\infty}|\cF^R_t]},
\]
where $E^y$ is expectation with respect to the product measure $P^y\times \nu$, where $P^y$ is the law of $R$ with $R_0=y$ and $\nu$ is the distribution of $F$. Since $F$ and $R$ are independent, under the assumption that $\bbP(E)= s(y)$, we obtain
\[
\frac{E^y\left[M_{\infty}\chf_{F}|\cF^R_t\right]}{E^y[M_{\infty}|\cF^R_t]}=\frac{\frac{s(R_t)}{s(y)}P(F)}{E^y[M_t|\cF^R_t]}=s(R_t),
\]
since  $P(F)=\bbP(E)=s(y)$ and $E^y[M_t|\cF^R_t]= 1$. Thus, a standard result from non-linear filtering (e.g. Theorem 8.1 in \cite{LS}) yields
\[
dY_t= dB^Y_t +\left\{rY_t +d \right\}dt
\]
for some $\cF^Y$-Brownian motion.
\end{proof}
We now turn to the computation of insider's value function.  We will in fact  verify in Proposition \ref{p:insoptber} that it is given by
\be \label{e:J}
J(x):= \int_{s^{-1}(v)}^x (s(y)-v)dy,
\ee
which satisfies (\ref{e:J'}) by construction for $h=s$. Moreover,
\bea
A J(x) -rJ(x)&=&\half s'(x) +(s(x)-v))\phi(x)-r\int_{s^{-1}(v)}^x (s(y)-v)dy \nn \\
&=&\half s'(x) +\int_{s^{-1}(v)}^x\phi(y)s'(y)dy \nn \\
&=&\half s'(x) -\half \int_{s^{-1}(v)}^x s''(y)dy\nn \\
&=&0, \label{e:AJ-rJ}
\eea
since $s'(s^{-1}(v))=0$. 
\begin{lemma} \label{l:Jexcessive} For any $v \in \{0,1\}$ consider the SDE
\be \label{e:optY}
Y_t = y +B_t +\int_0^t \left\{rY_s +d +\chf_{v =1} \frac{s'(Y_t)}{s(Y_t)} -\chf_{v =0}\frac{s'(Y_t)}{1-s(Y_t)}\right\}ds,
\ee
and define 

\[
A^v:= A + \left\{\chf_{v =1} \frac{s'}{s} -\chf_{v =0}\frac{s'}{1-s}\right\}\frac{d}{dx}.
\]
Then, $J$ is an $r$-excessive function for $A^v$ and $\bbE^{Q_v} e^{-rt}J(Y_t)\rar 0$ as $t \rar \infty$, where $\bbE^{Q_v}$ is the expectation operator with respect to the law of the solution of the above SDE for the given value of $v$.
\end{lemma}
\begin{proof}
Recall from (\ref{e:AJ-rJ}) that $AJ-rJ=0$ and observe that
\bean
A^v J -rJ&=& AJ -rJ + \left(\chf_{v=1} \frac{s'}{s}- r\chf_{v=0} \frac{s'}{1-s}\right)(s-v)\\
&=&\left(\chf_{v=1} \frac{s'}{s}- \chf_{v=0} \frac{s'}{1-s}\right)(s-v) \leq 0.
\eean
Define the non-negative function
\[
g_v:=\left(\chf_{v=1} \frac{s'}{s}- \chf_{v=0} \frac{s'}{1-s}\right)(v-s).
\]
$J$ will be $r$-excessive once we show that $g_v$ is integrable with respect to the speed measure of the diffusion defined by $A^v$. It follows from no. 31 on p. 34 of \cite{BorSal} that this measure is given by
\[
m^v(dx):=\left(\chf_{v=1}s^2(x) + \chf_{v=0}(1-s(x))^2\right)\frac{2}{ s'(x)}dx.
\]
Thus,
\[
\int_{-\infty}^{\infty}g_v(x)m^v(dx)= 2\int_{-\infty}^{\infty}s(x)(1-s(x))dx <\infty
\]
since $s$ is the cumulative distribution function of  a normal random variable.  Due to the integral representation formula for $r$-excessive functions (see no. 30 on p. 33 of \cite{BorSal}) we thus have
\[
J(x)= \int_{0}^{\infty} e^{-rs} Q_s^v g_v(x)ds + c_1 \psi_r(x)+ c_2\phi_r(x),
\]
where $(Q^v_t)_{t \geq 0}$ is the transition function for the solutions of (\ref{e:optY}) for a given value of $v$, and $\phi_r$ and $\psi_r$ are decreasing and increasing solutions of $A^vu-ru=0$, respectively. We claim that $c_1=c_2=0$. Indeed, suppose $v=0$. Then $J'(\infty)=1$ and $J(-\infty)=0$. However, since $l$ and $u$ are natural boundaries, $\phi_r(-\infty)=\infty=\psi'_r(\infty)$ (see p.19 of \cite{BorSal}), which in turn yields $c_1=c_2=0$. Similar considerations yield the same conclusion when $v=1$.

Therefore, 
\[
e^{-rt}Q^v_t J(y) = \int_{t}^{\infty} e^{-rs} Q_{s}^v g_v(y)ds.
\]
But the above converges to $0$ as $t \rar \infty$. Indeed, since
\[
Q_{t}^v g_v=\chf_{v=1}\frac{P_t sg_v}{s}+\chf_{v=0}\frac{P_t (1-s)g_v}{1-s},
\]
where $(P_t)_{t \geq 0}$ is the transition function for the Ornstein-Uhlenbeck process with the generator (\ref{e:OUgen}), we have
\[
Q_{t}^v g_v(y)\leq \max\left\{\frac{P_t s'(y)}{s(y)}, \frac{P_t s'(y)}{1-s(y)}\right\}\leq \frac{\sqrt{r}}{\sqrt{\pi}}\max\left\{\frac{1}{s(y)},\frac{1}{1-s(y)}\right\},
\]
\end{proof}
In view of the above lemma we can now construct the optimal strategy of the insider when the market makers use $s$ as their pricing rule.
\begin{proposition} \label{p:insoptber} Suppose that $P_t=s(Y_t)$, where $s$ is given by (\ref{e:OUscale}) and
\[
Y_t=y+ X_t +\int_0^t (rY_s +d)ds.
\]
with $y \in \bbR$ being chosen so that $s(y)=p=\bbP(\Gamma=1)$.
Then, the  following strategy is optimal for the insider:
\be \label{e:insiderstr}
\alpha_t=\left\{\Gamma \frac{s'(Y_t)}{s(Y_t)} -(1-\Gamma)\frac{s'(Y_t)}{1-s(Y_t)}\right\}.
\ee
The expected profit of the insider who uses this strategy equals $J(y)$, where $J$ is the function defined in (\ref{e:J}).
\end{proposition}
\begin{proof}
Let $J$ be given by (\ref{e:J}) and recall from (\ref{e:AJ-rJ}) that $AJ-rJ=0$. 
Thus,  Ito's formula together with integration by parts yield
\[
e^{-rt}J(Y_t)=J(y)+\int_0^t e^{-rs}(h(Y_s)-\Gamma)\{dB_s+\alpha_s ds\},
\]
where $\alpha$ is an arbitrary admissible strategy of the insider. 
Since $h$ and $v$ are bounded, 
\[
\int_0^t e^{-rs}(h(Y_s)-\Gamma)dB_s 
\]
is a uniformly integrable martingale. Thus, since $J\geq 0$, 
\[
\bbE^v\int_0^{\infty} e^{-rs}(v-h(Y_s))\alpha_sds \leq J(y),
\]
i.e., $J(y)$ is an upper bound for the expected profits of the insider. Observe that the distribution of $J(Y_t)$ under $\bbP^v$ depends on the choice of $\alpha$. Thus, if we can find an admissible $\alpha$ for which $\bbE^v \lim_{t\rar \infty}e^{-rt}J(Y_t)=0$, it will be optimal. However, Lemma \ref{l:Jexcessive} shows that if $\alpha$ is given by (\ref{e:insiderstr}), then this limit indeed equals $0$. Moreover, since $h$ is bounded, the admissibility of $\alpha$ will follow as soon as
\[
\int_0^t |\alpha_s|ds <\infty, \bbP^v\mbox{-a.s..}
\]

Indeed, on $[\Gamma=1]$ for every $T>0$, we have
\bean
\int_0^{T} \bbE^v\left\{\frac{s'(Y_t)}{s(Y_t)} \right\}dt&=&\int_0^{T} E^y\left\{\frac{s'(R_t)}{s(R_t)}s(R_t) \right\}dt\\
&=&\int_0^{T} E^y\left\{s'(R_t) \right\}dt,
\eean
in view of the absolute continuity relationship between the solutions of (\ref{e:optY}) and the Ornstein-Uhlenbeck process with the generator (\ref{e:OUgen}) as observed in the proof of Theorem \ref{t:optY}. Moreover, since $0<s<1$ and $s'$ is bounded the admissibility on the set $[\Gamma=1]$ is verified. It can be shown similarly that $\alpha$ is admissible on the set $[\Gamma=0]$. This completes the proof.
\end{proof}
A Markovian equilibrium for the market under consideration is given in the following:
\begin{theorem} \label{t:eq_bernoulli} Let $a \equiv 1$, $\phi(x)=rx+d$ for some $d\in \bbR$, $s$ be the function defined in (\ref{e:OUscale}), $y$ is the unique solution of $s(y)=p$, and $\alpha$ be the process given by (\ref{e:insiderstr}). Then, $((1, \phi, y), s, \alpha)$ is an equilibrium.
\end{theorem}
\begin{proof}
Given this choice of $\sigma$, $\phi$, and $s$, we have seen in Proposition \ref{p:insoptber} that $\alpha$ is optimal. Thus, it remains to show, in view of Proposition \ref{p:PV}, that 
\[
s(Y_t)=\bbE[\Gamma|\cF^X_t]=\bbE[\Gamma|\cF^Y_t]
\]
as all $\cF^X$ (or, equivalently, $\cF^Y$) martingales are continuous due to the fact the filtration is Brownian (see Corollary \ref{c:incons} for another manifestation of this fact). 
However, this assertion follows easily from Theorem \ref{t:optY} since $s(y)=\bbP(\Gamma=1)$.
\end{proof}

Although the theorem above gives the impression that there is a continuum of equilibria indexed by $d$, in fact all these choices of $\phi$ lead to the same SDE for the price. Moreover, the insider's optimal strategy does not depend on $d$ when expressed in terms of $P$.
\begin{corollary} \label{c:eq_bernoulli}Let $d$, $\phi$, and $\alpha$ be as in Theorem \ref{t:eq_bernoulli} and let $P^*$ be the equilibrium price associated to $((1, \phi, y), s, \alpha)$. Then, $P^*$ is a uniformly integrable $\cF^X$-martingale with $P^*_{\infty}=\Gamma$, and
\be \label{e:P*dyn}
P^*_t =\bbE[\Gamma] + \int_0^t \lambda(P^*_s)\left\{dB_s + \left( \Gamma \frac{\lambda(P^*_s)}{P^*_s}-(1-\Gamma)\frac{\lambda(P^*_s)}{1-P^*_s}\right)ds\right\},
\ee
where $\lambda(x)=s'_0(s_0^{-1}(x))$ and $s_0$ is the function in (\ref{e:OUscale}) with $d=0$. In particular, the function $\lambda$ does not depend on $d$.

Moreover, the insider's strategy in the equilibrium has the following from:
\be \label{e:insstrP}
\alpha_t^*=\Gamma \frac{\lambda(P^*_t)}{P^*_t}-(1-\Gamma)\frac{\lambda(P^*_t)}{1-P^*_t}.
\ee
\end{corollary}
\begin{proof} Let $d \in \bbR$ be fixed and $Y=B+\int_0^{\cdot}(\phi(Y_t)+\alpha_t)dt$ so that $P^*_t=s(Y_t)$, where $s$ is given by (\ref{e:OUscale}). Then, 
\[
dP^*_t= s'(Y_t)\left\{dB_t + \left(\Gamma \frac{s'(Y_t)}{s(Y_t)} -(1-\Gamma)\frac{s'(Y_t)}{1-s(Y_t)}\right)dt\right\}.
\]
Observe that $P^*_0=s(y)=\bbE[\Gamma]$ by the choice of $y$. 

Next, note that
\[
s(x)=s_0\left(x+\frac{d}{r}\right),
\]
implying $s'(x)=s_0'(x+\frac{d}{r})$ as well as $s^{-1}(x)=s_0^{-1}(x)-\frac{d}{r}$. Combining these two observations we then deduce that
\[
s'(s^{-1}(x))=s_0'(s_0^{-1}(x)),\; \forall d \in \bbR.
\]
Therefore,
\[
dP^*_t= s_0'(s_0^{-1}(P^*_t))\left\{dB_t + \left(\Gamma \frac{ s_0'(s_0^{-1}(P^*_t))}{P^*_t} -(1-\Gamma)\frac{ s_0'(s_0^{-1}(P^*_t))}{1-P^*_t}\right)dt\right\}.
\]
This yields, in view of the definition of $\lambda$, the dynamics of $P^*$ given by (\ref{e:P*dyn}).

That $P^*$ is uniformly integrable follows from the boundedness of $s$. Its limiting property $P^*_{\infty}=s(Y_{\infty})=\Gamma$ is due to  Theorem \ref{t:optY}.

The form of the insider's strategy follows immediately from the corresponding change of variable.
\end{proof}

In Kyle's model with risk-neutral market makers it is in general observed that in equilibrium the insider's trades are inconspicuous, i.e. the distribution of the equilibrium demand process equals to that of the noise trades. We observe the same phenomenon here.
\begin{corollary} \label{c:incons} Consider the equilibrium described in Theorem \ref{t:eq_bernoulli} or Corollary \ref{c:eq_bernoulli} and let $X^*$ denote the equilibrium level of demand. Then, $X^*$ is a Brownian motion in its own filtration. 
\end{corollary}
\begin{proof}
Note that 
\[
X^*_t=B_t +\int_0^t \alpha^*_sds
\]
where $\alpha^*$ is given by (\ref{e:insstrP}). Recall that the relationship between $X$ and $Y$ entails they generate the same filtration. Moreover, since the pricing rule is a strictly increasing function, we deduce that the filtrations generated by $P^*$ and $X^*$ coincide. Therefore,
\bean
\bbE\left[\alpha^*_t\big|\cF^{X^*}_t\right]&=&\bbE\left[\alpha^*_t\big|\cF^{P^*}_t\right]\\
&=&\bbE\left[\Gamma\big|\cF^{P^*_t}\right] \frac{\lambda(P_t)}{P_t}-\left(1-\bbE\left[\Gamma\big|\cF^{P^*_t}\right]\right)\frac{\lambda(P_t)}{1-P_t}\\
&=&0
\eean
since $P^*$ is a uniformly integrable $\cF^{Y^*}$-martingale with $P^*_{\infty}=\Gamma$ by Corolloary \ref{c:eq_bernoulli}. Thus,  $X^*$ is a Brownian motion in its own filtration by Theorem 8.1 in \cite{LS}.
\end{proof}
Since we are able to characterise the equilibirium price process as a function of an Ornstein-Uhlenbeck process, this allows us to observe a deviation in equilibrium from the original model in Kyle. In \cite{K}, as explained in \cite{bcw} and \cite{CD}, the reciprocal of market depth, follows a martingale preventing `systematic changes' in the market depth. In the absence of such systematic changes the insider cannot acquire a large position when the depth is low to liquidate at a later date when the liquidity is higher to obtain unbounded profits. The reciprocal of the market depth is called {\em Kyle's lambda} and it is given by the process $\lambda(P^*)$ in our model, where $\lambda$ and $P^*$ are as in Corollary \ref{c:eq_bernoulli}. The next result shows that, contrary to Kyle's findings, the reciprocal of the market depth follows a supermartingale, which was also observed by Back and Baruch in \cite{BB} using different arguments.
\begin{corollary} \label{c:depth} Let $\lambda$ and $P^*$ be as in  Corollary \ref{c:eq_bernoulli}. Then, $\lambda(P^*)$ is an $\cG^M$-supermartingale such that $\lim_{t \rar \infty}\bbE\left[\lambda(P^*_t)\right]=0$, i.e. $\lambda(P^*)$ is an $\cG^M$-potential.
\end{corollary}
\begin{proof}
As observed earlier suppose without loss of generality that $b=0$ so that $\phi(x)=rx$ and the market makers' signal, $Y^*$, solves in its own filtration
\[
dY^*_t= d\beta_t +rY^*_tdt,
\]
where $\beta$ is an $\cF^{Y^*}$-Brownian motion. Since $P^*=s_0(Y^*)$, it suffices to show that $s_0'(Y^*)$ is a $\cF^{Y^*}$-supermartingale. Independence of $\tau$ and $P^*$ will then imply that $s_0'(Y^*)$, hence $\lambda(Y^*)$, is a $\cG^M$-supermartingale.

On the other hand,
\[
s_0'(x)=\sqrt{\frac{r}{\pi}}e^{-rx^2}.
\]
Thus, an application of Ito's formula yields
\[
ds_0'(Y^*_t)=-s_0'(Y^*_t)\left(2rY^*_td\beta_t +rdt\right),
\]
which in turn implies $s_0'(Y^*)$ has the required supermartingale property. Moreover, as a consequence of the Dominated Convergence Theorem, we obtain
\[
\lim_{t \rar \infty} \bbE\left[s_0'(Y^*_t)\right]=\bbE\left[\lim_{t \rar \infty}s_0'(Y^*_t)\right]=0
\]
since $\lim_{t \rar \infty}|Y^*_t|=\infty$.
\end{proof}
\begin{remark} In fact the actual marginal price impact that is observed in the market is given by $\lambda(P_t)\chf_{[\tau>t]}$ since the price is constant from $\tau$ onwards. Note that this is again a $\cG^M$-supermartingale in view of Lemma \ref{l:pen}.
\end{remark}
	
A simple consequence of the above result is that the market gets more liquid on average as time passes. The reason for this deviation from Kyle is due to the random deadline to the whole trading activities, which is independent of everything else.  The insider is aware of the fact that her informational advantage is going to end at a totally inaccessible stopping time, which will come as a surprise. She nevertheless chooses not to trade aggressively revealing her information quickly since, as observed by Back and Baruch \cite{BB}, the price impact is decreasing in time so the risk of waiting can be compensated by lower execution costs.
\begin{remark} \label{r:extension} The computations made in this section suggest that there is no equilibrium unless $\Gamma$ has a Bernoulli distribution. Indeed, if $X$ is independent of $\tau$, Proposition \ref{p:YOU} implies that $Y$ is a transient Ornstein-Uhlenbeck process. On the other hand, $P_{\infty}=\Gamma$ and the equilibrium price is given by a scale function of $Y$. Consequently, $P_{\infty}$ can take only two different values since $Y_{\infty}\in \{-\infty, \infty\}$.
\end{remark} 
\subsection{More general liquidation value} \label{s:general}  In this section we consider more general distributions for $\Gamma$ and suppose that $\Gamma \eid f(\eta)$ for some {\em continuous and strictly increasing} $f$ and a standard Normal random variable, $\eta$. We can incorporate atoms and consider more general distributions for $\Gamma$. However, this will only result in more complicated coefficients for $Y^*$ in the filtration of the insider and will not significantly alter the qualitative inferences that one can make within this model. Thus, for the simplicity of  the exposition and the model we make the following assumption.
\begin{assumption} \label{a:noatom}
There exists a continuous and strictly increasing function $f:\bbR\mapsto \bbR$ such that $\Gamma\eid f(\eta)$, where $\eta$ is a standard normal random variable\footnote{$\eid$ stands for {\em equality in distribution}.}.
\end{assumption}
Remark {\ref{r:extension} suggests that one cannot go beyond a Bernoulli distributed payoff  using a time-homogeneous SDE for $Y$. Collin-Dufresne et al. \cite{CFM}  consider a similar problem when $\Gamma\eid\eta$ and obtain an equilibrium, where the coefficients of $Y^*$ depends on time using ideas from Kalman filtering. 

 We will next show that an equilibrium exists for more general payoff  distributions. In fact, the market makers' equilibrium signal process, $Y^*$, will turn out to be just a time-change of the Ornstein-Uhlenbeck process that appears in the equilibrium when $\Gamma$ had a Bernoulli distribution as in last section. 

To wit, let's suppose
\[
dY_t =\sigma(t)a(Y_t)dX_t + \sigma^2(t)\phi(Y_t)dt,
\]
where $\sigma:\bbR_+ \mapsto \bbR_+$. Following similar arguments that were used in the beginning of Section \ref{s:Bernoulli} after the obvious modifications we obtain
\bea
h_t + \frac{1}{2}a^2 \sigma^2 h_{yy} + \sigma^2 \phi h_y&=& 0 \label{e:ghrule} \\
\frac{1}{2}a\sigma^2 a'' +\phi\sigma^2 \frac{a'}{a} -\sigma^2 \phi' +\frac{\sigma'}{\sigma} +r&=& 0. \label{e:gphia}
\eea
It is easy to check that when $\sigma \equiv 1$ the above equations reduce to (\ref{e:hrule}) and (\ref{e:phia}). 

As in previous section we shall choose $a\equiv 1$ and $\phi(x)=rx$. This implies
\be \label{e:ODEsigma}
\frac{\sigma'}{\sigma(1-\sigma)(1+\sigma)}=-r.
\ee
Since 
\[
\frac{1}{x(1-x)(1+x)}=\frac{1}{x}+\frac{1}{2(1-x)}-\frac{1}{2(1+x)},
\]
integrating (\ref{e:ODEsigma}) yields
\[
\frac{\sigma^2(t)}{1-\sigma^2(t)}=C^2 e^{-2rt}
\]
for some constant $C$ to be determined later. Thus, we can solve the above to deduce
\be \label{e:sigma}
\sigma^2(t)= \frac{C^2  e^{-2rt}}{1+C^2 e^{-2rt}}.
\ee
The next lemma will define a function, which will later turn out to be the value function for the insider. 
\begin{lemma} \label{l:gJ} Let $h:\bbR_+\times \bbR \mapsto \bbR$ satisfy (\ref{e:ghrule}) such that $h(t,\cdot)$ is strictly increasing for each $t \geq 0$. Consider
\be \label{e:gJ}
J(t,y):=\int_{h^{-1}(t,v)}^y \frac{h(t,x)-v}{\sigma(t)}dx +\half e^{rt}\int_t^{\infty}e^{-rs}\sigma(s)h_y(s,h^{-1}(s,v))ds,
\ee
where $h^{-1}(t,\cdot)$ represents the inverse of $h(t,\cdot)$ for every $t \geq 0$.
Then,
\[
J_t +\half \sigma^2(t)J_{yy}+\sigma^2(t)ryJ_y -rJ=0
\]
provided
\[
\int_t^{\infty}e^{-rs}\sigma(s)h_y(s,h^{-1}(s,v))ds<\infty, \; \forall t\geq 0.
\]
\end{lemma}
\begin{proof}
Direct differentiation yields
\bean
J_t&=&\int_{h^{-1}(t,v)}^y \frac{h_t(t,x)}{\sigma(t)}dx - \int_{h^{-1}(t,v)}^y \frac{h(t,x)-v}{\sigma^2(t)}\sigma'(t)dx\\
&&+\frac{re^{rt}}{2}\int_t^{\infty}e^{-rs}\sigma(s)h_y(s,h^{-1}(s,v))ds-\half\sigma(s)h_y(t,h^{-1}(t,v))\\
&=&-\sigma(t)\int_{h^{-1}(t,v)}^y \left\{\half h_{yy}(t,x)+ h_y(t,x) rx \right\}dx- \int_{h^{-1}(t,v)}^y \frac{h(t,x)-v}{\sigma^2(t)}\sigma'(t)dx\\
&&+\frac{re^{rt}}{2}\int_t^{\infty}e^{-rs}\sigma(s)h_y(s,h^{-1}(s,v))ds-\half\sigma(s)h_y(t,h^{-1}(t,v)) \\
&=&\frac{\sigma(t)}{2}\left(h_y(t,h^{-1}(t,v))-h_y(t,y)\right) - \int_{h^{-1}(t,v)}^y \frac{h(t,x)-v}{\sigma^2(t)}\sigma'(t)dx\\
&&-\sigma(t)\int_{h^{-1}(t,v)}^y  h_y(t,x) rx dx+\frac{re^{rt}}{2}\int_t^{\infty}e^{-rs}\sigma(s)h_y(s,h^{-1}(s,v))ds-\half\sigma(t)h_y(t,h^{-1}(t,v)).
\eean
Thus, 
\bean
J_t +AJ -rJ&=&\sigma(t)ry (h(t,y)-v) -\int_{h^{-1}(t,v)}^y \frac{(h(t,x)-v)(r\sigma+\sigma'(t))}{\sigma(t)}dx\\
&&-\sigma(t)\int_{h^{-1}(t,v)}^y  h_y(t,x) rx dx\\
&=&\sigma(t)ry (h(t,y)-v) -r\sigma(t)\int_{h^{-1}(t,v)}^y (h(t,x)-v)dx-\sigma(t)\int_{h^{-1}(t,v)}^y  h_y(t,x) rx dx\\
&=&0,
\eean
where the last equality follows from integration by parts.
\end{proof}
The PDE (\ref{e:ghrule}) satisfied by $h$ indicates that the market makers's signal in equilibrium will be a time-changed Ornstein-Uhlenbeck process, where the time change is given by
\be \label{e:timechange}
V(t):=\int_0^t\sigma^2(s)ds=\frac{1}{2r}\log \frac{1+C^2}{1+C^2 e^{-2rt}}.
\ee
Indeed, any solution of (\ref{e:ghrule}) can be obtained by a time change as we see in the next lemma.
\begin{lemma} Suppose that $a\equiv 1$ and $\phi(x)=rx$. Then $h$ is a solution of (\ref{e:ghrule}) iff $g$ defined by
\[
g(t,y):=h(V^{-1}(t),y),
\]
where $V$ is the absolutely continuous function given in (\ref{e:timechange}), solves
\be \label{e:gpde}
g_t+\half g_{yy}+ ry g_y=0.
\ee
\end{lemma}
\begin{proof}
Note that $h(t,y)=g(V(t),y)$. Since $dV(t)=\sigma^2(t)dt$, the claim follows.
\end{proof}
Consistent with the findings of the previous section the insider should construct a bridge process, $Y^*$, such that $\lim_{t \rar \infty} h^*(t,Y^*_t)=\Gamma$ and $Y^*$, in its own filtration, follows
\[
dY^*_t= \sigma(t) dB^Y_t + r\sigma^2(t)Y^*_tdt,
\]
i.e. a time-changed version of the Ornstein-Uhlenbeck process with the generator (\ref{e:OUgen}), where $d=0$. As we observed in the previous lemma the time change is given by the function $V(t)$, which converges to $V(\infty)=\frac{1}{2r}\log (1+C^2)<\infty$. This implies that the distribution of $Y^*_{\infty}$ equals that of the Ornstein-Uhlenbeck process at $V(\infty)$ defined by (\ref{e:OUgen}) with $d=0$. This distribution is Gaussian and allows us to go beyond a Bernoulli distribution for $\Gamma$. 

Let $p(t,x,y)$ be the transition density of the Ornstein-Uhlenbeck process defined by (\ref{e:OUgen}) with $d=0$. It is well-known that
\[
p(t,x,y)= q\left(\frac{e^{2rt}-1}{2r}, y-x e^{rt}\right),
\]
where $q(t,x)=\frac{1}{\sqrt{2 \pi t}}e^{-\frac{x^2}{2t}}$. The next theorem defines the bridge process that will be the key to the insider's strategy in the equilibrium.
\begin{theorem} \label{t:gbridge} Let $f$ be the function in Assumption \ref{a:noatom} and assume $C^2=2r$. Then, for any $v \in \bbR$, there exists a unique strong solution to 
\be
Y_t =\int_0^t\sigma(s)dB_s + r \int_0^t \frac{f^{-1}(v)- Y_s \cosh\left(\half \log (1+ 2r e^{-2rs})\right)}{\sinh\left(\half \log (1+ 2r e^{-2rs})\right)}\sigma^2(s)ds,
\ee
where $\sigma>0$ is defined via (\ref{e:sigma}). Moreover,
$\lim_{t \rar \infty} Y_t=f^{-1}(v)$, $\bbQ^v$-a.s., where $\bbQ^v$ is the law of the solution. Moreover,
\be \label{e:genYMM}
\bbE^{\bbQ^v}\left[F(Y_s; s\leq t)\right]=\frac{\bbE^{\bbQ}\left[p(V(\infty)-V(t), Y_t, f^{-1}(v))F(Y_s; s\leq t)\right]}{p(V(\infty),0, f^{-1}(v))},
\ee
where $F$ is a bounded measurable function and $\bbQ$ is the law of the unique solution to 
\be \label{e:genYtc}
Y_t= \int_0^t\sigma(s)dB_s + r \int_0^t \sigma^2(s)Y_s ds.
\ee
\end{theorem}
\begin{proof}
The existence and uniqueness of the solution follows immediately since the SDE has Lipschitz coefficients in every compact interval $[0,T]$.  

Next observe that
\[
\half \log (1+ 2r e^{-2rt})=r(V(\infty)-V(t)).
\]
Thus, if we define $R_t:=Y_{V^{-1}(t)}$, we obtain
\[
 R_t =\int_0^{V^{-1}(t)}\sigma(s)dB_s + r \int_0^t \frac{f^{-1}(v)- R_s \cosh\left(r(V(\infty)-s)\right)}{\sinh\left(r(V(\infty)-s)\right)}ds.
 \]
 On the other hand, $\beta_t:=\int_0^{V^{-1}(t)}\sigma(s)dB_s$ is a local martingale with respect to the filtration $(\cG_t)_{t \geq 0}$, where $\cG_t:=\cF_{V^{-1}(t)}$. Moreover, $[\beta,\beta]_t=t$ for each $t\geq 0$. Therefore, it is a $\cG$-Brownian motion. Consequently,
 \be \label{e:TCY}
 R_t =\beta_t + r \int_0^t \frac{f^{-1}(v)- R_s \cosh\left(r(V(\infty)-s)\right)}{\sinh\left(r(V(\infty)-s)\right)}ds, \qquad t <V(\infty).
 \ee
 However, the above is the SDE for 
 \[
 \rho_t= W_t +r \int_0^t \rho_s ds
 \]
 conditioned on the event $[\rho_{V(\infty)}=f^{-1}(v)]$. Indeed, the SDE representation of this Markovian bridge follows from Example 2.3 in \cite{CD2}, which coincides with the above SDE since $F(t)=e^{rt}$ and $\Sigma(s,t)=\frac{e^{2rt}-1}{2r}$, where $F$ and $\Sigma$ are the functions defined in Example 2.3 of \cite{CD2}. Therefore, $R_t \rar f^{-1}(v)$ as $t \rar V(\infty)$, which is equivalent to $Y_t \rar f^{-1}(v)$ as $t \rar \infty$.
 
 The absolute continuity relationship is a consequence of Theorem 2.2 in \cite{CD2} since  the solution of (\ref{e:genYtc}) is a Markov process  with transition density $p(V(t)-V(s),y,z)$.
\end{proof}
 We are now ready to state the existence of an equilibrium in the next theorem, whose proof is postponed to Appendix \ref{a:t}.
\begin{theorem} \label{t:geneq} Suppose $C^2= 2r$ and  assume that  there exist positive constants $K>0$ and $k<\frac{1}{1+2r}$ such that 
\[
|f(y)|\leq K e^{\frac{ky^2}{4}}.
\]
Define 
\[
g(t,y):=\int_{-\infty}^{\infty}f(z) p(V(\infty)-t,y,z)dz
\]
and let $h^*(t,y)=g(V(t),y)$. Then, $((\sigma^*, \phi^*,0), h^*, \alpha^*)$ is an equilibrium, where $\sigma^*$ is the positive square root of (\ref{e:sigma}) with $C^2=2r$, $\phi^*(y)=ry$, and 
\[
\alpha^*(t)=  r \sigma^*(t)\left(\frac{f^{-1}(\Gamma)- Y^*_t \cosh\left(\half \log (1+ 2r e^{-2rt})\right)}{\sinh\left(\half \log (1+ 2r e^{-2rt})\right)}- Y^*_t\right).
\]
Moreover, 
\be \label{e:Y*proj}
Y^*_t= \int_0^t \sigma^*(s)dB^*_s +r\int_0^t(\sigma^*(s))^2Y^*_s ds,
\ee
where $B^*$ is an $\cF^{Y^*}$-Brownian motion.
\end{theorem}

As in the time-homogeneous case, Kyle's lambda will be a uniformly integrable supermartingale. However, contrary to the time-homogeneous case it won't disappear as $t \rar \infty$, i.e. it won't be a potential in general.
\begin{proposition} Consider the equilibrium given in Theorem \ref{t:geneq} and define $\lambda^*_t:=h^*_y(t,Y^*_t)$. Then, $\lambda^*$ is a uniformly integrable $(\cF^{Y^*}, \bbP)$-supermartingale and
\be \label{e:limitimpact}
\lim_{t \rar \infty}\bbE[\lambda^*_t]=\int_{-\infty}^{\infty}f'(z)\frac{1}{\sqrt{2\pi}}e^{-\frac{z^2}{2}}dz.
\ee
\end{proposition}
\begin{proof}
Define $u(t,y)=h^*_y(t,y)$. Differentiating 
\[
h^*_t + \frac{\sigma^2(t)}{2}h^*_{yy} +\sigma^2(t)ryh^*_y=0
\]
with respect to $y$ yields
\[
u_t + \frac{\sigma^2(t)}{2}u_{yy} +\sigma^2(t)ryu_y=-r\sigma^2(t) u.
\]
Applying Ito formula to $u$ and $Y^*$, which satisfies (\ref{e:Y*proj}), yields
\[
d\lambda^*_t =u_y(t,Y^*_t)\sigma^*(t)dB^*_t -r(\sigma^*(t))^2\lambda^*_tdt.
\]
Since $\lambda^*\geq 0$, the stochastic integral in the above decomposition is a supermartingale, which leads to the desired supermartingale property of $\lambda^*$. 

It follows from (\ref{e:h*fk}) that
\[
u(t,y)=\sqrt{1+2re^{-2rt}}\int_{-\infty}^{\infty} f'(z)p(V(\infty)-V(t),y,z)dz.
\]
Thus, $\lambda^*_{\infty}:=\lim_{t \rar \infty}\lambda^*_t=f'(Y^*_{\infty}), \, \bbP$-a.s. and 
\[
\bbE[u(t,Y^*_t)]=\sqrt{1+2re^{-2rt}}\int_{-\infty}^{\infty} f'(z)p(V(\infty),0,z)dz
\]
by Chapman-Kolmogorov equation. Therefore,
\[
\lim_{t \rar \infty} \bbE[u(t,Y^*_t)]=\int_{-\infty}^{\infty} f'(z)p(V(\infty),0,z)dz=\bbE[f'(Y^*_{\infty})]=\bbE[\lambda^*_{\infty}].
\]
This implies $\lambda^*$ is a uniformly integrable supermartingale. (\ref{e:limitimpact}) follows from the fact that $Y^*_{\infty}$ is standard normal.
\end{proof}
\begin{remark} As in Corollary \ref{c:incons} one can show that the insider's trades are inconspicuous, i.e. $X^*$ is a Brownian motion in its own filtration. This directly follows from the equilibrium level of $Y^*$, which satisfies (\ref{e:Y*proj}).
\end{remark}
\section{Conclusion} \label{s:conc}
Using tools from potential theory of one-dimensional diffusions we have solved a version of the Kyle model with general payoffs when the announcement date has an exponential distribution and is independent of all other parameters of the model. It is shown that a stationary equilibrium exists only if the payoff has a Bernoulli distribution, which corresponds to the special case considered in \cite{BB}. The approach considered herein is novel in its study and characterisation of the optimal strategies of the insider in terms of excessive functions of an associated diffusion process. In particular the so-called Kyle's lambda, which is a measurement of liquidity, is identified with a {\em potential}. 

As in the earlier literature on the Kyle model we have shown that the total demand for the asset in equilibrium has the same distribution as that of the noise traders, i.e. the insider's trades are inconspicuous. What is different from the earlier models, however, is that the equilibrium prices no longer converge to the payoff, $\Gamma$, as the time approaches to the announcement date. That is, there is a jump in the price when $\Gamma$ becomes public knowledge. This is due to the fact that the announcement comes as a surprise even for the insider and, therefore, she is not able to construct a bridge of random length $\tau$ for the demand process in order for the prices to converge to $\Gamma$ (cf. bridge construction in \cite{GBP}). However, in equilibrium she trades in such a way that the price process conditioned on {\em no announcement}, i.e. $P$, converges to $\Gamma$.

It will be interesting to see how these conclusions, especially the last one, change if the announcement date, $\tau$, is no longer assumed to be independent of the other variables. However, this would require a framework beyond the scope of the current paper and is left for future study.

\appendix

\section{Proof of Proposition \ref{p:PV}} \label{a:p}
Recall that we are searching for a decomposition 
\[
S_t= P_t \chf_{[t<\tau]} + \Gamma \chf_{[t \geq \tau]},
\]
where $P$ is a semimartingale adapted to $\cF^X$.
 In order for $P$ to be a candidate price process on the set $[t<\tau]$ one  must have that $S$ is a $\cG^M$-martingale. To this end the following lemma will be crucial. 

\begin{lemma} \label{l:pen} Define
	\bean
	N_t&:=& \Gamma \chf_{[t \geq \tau]}-r\int_0^t \chf_{[s<\tau]}\bbE[\Gamma|\cF^X_s]ds,\\
	M_t&:=&\chf_{[t \geq \tau]}-r\int_0^t \chf_{[s<\tau]}ds
	\eean
	Then, $N$ and $M$ are $\cG^M$-martingales. 
\end{lemma}
\begin{proof}
	Note that for $s<t$ 
	\bean
	\bbE[N_t|\cG^M_s]&=&\chf_{[\tau\leq s]}\bbE\left[\Gamma-r\int_0^{\tau}\bbE[\Gamma|\cF^X_u]du\bigg|\cG^M_s\right]+\chf_{[\tau>s]}\bbE\left[\Gamma \chf_{[t \geq \tau]}-r\int_0^t \chf_{[u<\tau]}\bbE[\Gamma|\cF^X_u]du\bigg|\cG^M_s\right]\\
	&=&\chf_{[\tau\leq s]}N_s +\chf_{[\tau>s]}\left(\bbE[\Gamma|\cF^X_s](1-e^{-r(t-s)})-r\int_0^s \bbE[\Gamma|\cF^X_u]du\right)\\
	&&-\chf_{[\tau>s]}\bbE\left[\int_s^t r \chf_{[u<\tau]}\bbE[\Gamma|\cF^X_u]du\bigg|\cG^M_s\right]\\
	&=&\chf_{[\tau\leq s]}N_s +\chf_{[\tau>s]}\left(\bbE[\Gamma|\cF^X_s](1-e^{-r(t-s)})-r\int_0^s \bbE[\Gamma|\cF^X_u]du\right)\\
	&&-\chf_{[\tau>s]}\int_s^t r e^{-r(u-s)}\bbE[\Gamma|\cF^X_s]du\\
	&=&N_s,
	\eean
	where the second and third equalities are due to the independence of $X$ and $\tau$. 
	
	The proof for the martingale property of $M$ follows the similar lines.
\end{proof}

In view of the above lemma in order for $S$ to be a $\cG^M$-martingale one needs to show that
\be \label{def:U}
U_t:=P_t \chf_{[t<\tau]} + r\int_0^t \chf_{[s<\tau]}\bbE[\Gamma|\cF^X_s]ds
\ee
is a $\cG^M$-martingale. This leads to 

\begin{proof}[Proof of Proposition \ref{p:PV}] Suppose that the semimartingale decomposition of $P$ is given by $P_t=P_0+Z_t +A_t$ where $Z$ is a continuous local martingale and $A$ is a predictable process of finite variation\footnote{Under the assumption that all martingales are continuous the optional and the predictable $\sigma$-algebras coincide.}. Let 
	\[
	U_t=P_t \chf_{[t<\tau]} + r\int_0^t \chf_{[s<\tau]}\bbE[\Gamma|\cF^X_s]ds.
	\]
	Then, in view of the previous lemma we have
	\bean
	dU_t&=&\chf_{[t \leq \tau]}\{dZ_t +dA_t\} -P_{t-}\{dM_t +r \chf_{[t<\tau]}dt\} + r \chf_{[t<\tau]}\bbE[\Gamma|\cF^X_t]dt -\Delta A_{\tau} \chf_{[t =\tau]}\\
	&=&\chf_{[t \leq \tau]}dZ_t-P_{t-}dM_t+ \chf_{[t <\tau]}dA_t +r\chf_{[t<\tau]}\left(\bbE[\Gamma|\cF^X_t]-P_t\right)dt.
	\eean
	Consequently,
	\[
	\int_0^t \left\{\chf_{[s<\tau]}dA_s +r\chf_{[s<\tau]}\left(\bbE[\Gamma|\cF^X_s]-P_s\right)ds\right\}
	\]
	is a predictable local martingale of finite variation. Thus, 
	\[
	\chf_{[t <\tau]}dA_t=\chf_{[t <\tau]}r(P_t -\bbE[\Gamma|\cF^X_t])dt.
	\]
	Since $\tau$ is independent of $\cF^X$ and $\bbP(\tau>t)>0$ for all $t\geq 0$, this yields
	\[
	dA_t=r(P_t -\bbE[\Gamma|\cF^X_t])dt.
	\]
	Thus, $A$ is continuous and so is $P$. Consider 
	\[
	\tau_n :=\inf\{t\geq 0: |Z_t|>n\}.
	\]
	Since $Z$ is continuous, $\tau_n \rar \infty$, $\bbP$-a.s..  
	
	Next, let $\hat{\Gamma}_t=\bbE[\Gamma|\cF^X_t]$ and consider $\sigma_n:=\inf\{t\geq 0: |P_t-\hat{\Gamma}_t|>n\}$, which converges to $\infty$ due to the continuity of $P$ and $\hat{\Gamma}$.  Then, 
	\[
	P_{t\wedge \tau_n\wedge \sigma_m}-\hat{\Gamma}_{t\wedge\tau_n\wedge \sigma_m} -r \int_0^{t\wedge \tau_n\wedge \sigma_m} (P_s-\hat{\Gamma}_s)ds
	\]
	is a  martingale for each $n$ and and $m$. Moreover, denoting the semimartingale local time of $P-\hat{\Gamma}$ at $0$ by $L$, we deduce that
	\be \label{e:P-Gdecomp}
	\left|P_{t\wedge \tau_n\wedge \sigma_m}-\hat{\Gamma}_{t\wedge \tau_n\wedge \sigma_m}\right|- r \int_0^{t\wedge \tau_n\wedge \sigma_m} \left|P_u-\hat{\Gamma}_u\right|du- {L}_{t\wedge \tau_n\wedge \sigma_m}
	\ee
	is a $\cG$-local martingale and, thus a submartingale being bounded from above. 
	
	Also, note that since $S$ is a uniformly integrable martingale as observed on the proof of Proposition \ref{p:SUI}, $\chf_{[\tau>t]}|P_t|\leq S_t$, and $\tau$ is independent of $P$, $(P_{t\wedge \tau_n\wedge \sigma_m})_{ n\geq1, m \geq 1}$ is uniformly integrable. Thus,  we obtain for $s<t$
	\[
	\bbE\left[|P_t-\hat{\Gamma}_t|\big| \cF^X_s\right] \geq  |P_s-\hat{\Gamma}_s| + r \int_s^t \bbE\left[|P_u-\hat{\Gamma}_u|\big|\cF^X_s\right]du \]
	after taking limits as $n \rar \infty $ and $m \rar \infty$ and utilising the monotone convergence theorem on the integral in (\ref{e:P-Gdecomp}) as well as the fact that ${L}$ is increasing.
	
	A straightforward application of Gronwall's inequality, therefore,  implies that for any $t >s$
	\be \label{e:gronwall}
	\bbE\left[\left|P_t-\hat{\Gamma}_t\right|\big| \cF^X_s\right] \geq \left|P_s-\hat{\Gamma}_s\right|e^{r(t-s)}.
	\ee
	On the other hand,
	\bean
	\chf_{[\tau>s]}\bbE\left[\left|S_t-\Gamma\right| \big| \cG^M_s\right]&=&\chf_{[\tau>s]}\bbE\left[\chf_{\tau>t}\left|P_t-\Gamma\right| \big| \cG^M_s\right]\\
	&=&\chf_{[\tau>s]}e^{-r(t-s)}\bbE\left[\left|P_t-\Gamma\right| \big| \cF^X_s\right]\geq \chf_{[\tau>s]}e^{-r(t-s)}\bbE\left[\left|P_t-\hat{\Gamma}_t\right| \big| \cF^X_s\right],
	\eean
	where the last inequality follows from Jensen's inequality since $P_t$ is $\cF^X_t$-measurable. 
	
	However, (\ref{e:gronwall}) then yields
	\[
	\lim_{t \rar \infty} \chf_{[\tau>s]}\bbE\left[\left|S_t-\Gamma\right|\big| \cG^M_s\right]\geq \lim_{t \rar \infty}\left|P_s-\hat{\Gamma}_s\right|=\left|P_s-\hat{\Gamma}_s\right|,
	\]
	which contradicts (\ref{e:S-Glimit}) unless $P_s =\hat{\Gamma}_s$, $\bbP$-a.s.. Since $P$ and $\hat{\Gamma}$ are continuous, the null set can be chosen to be independent of $s$, which completes the proof.
\end{proof}
\section{Proof of Theorem \ref{t:geneq}} \label{a:t}
	We will show that $((\sigma^*, \phi^*,0), h^*, \alpha^*)$ is an equilibrium by checking 1) $\alpha^*$ is admissible and optimal given $(\sigma^*, \phi^*,0)$ and 2) $(\sigma^*, \phi^*,0)$ is an admissible pricing rule given $\alpha^*$.
	\begin{itemize}
		\item[Step 1.] {\em Insider's optimality.} In view of Lemma \ref{l:gJ} let us first verify that 
		\[
		\int_t^{\infty}e^{-rs}\sigma^*(s)h^*_y(s,h^{-1}(s,v))ds<\infty, \; \forall t\geq 0.
		\]
		To see this first observe that $h^*_y(s,h^{-1}(s,v))=\frac{1}{\frac{d h^{-1}(s,y)}{dy}}$ and 
		\bea
		y&=&\int_{-\infty}^{\infty}f(z)p(V(\infty)-V(t),h^{-1}(t,y),z)dz \nn \\
		&=&\int_{-\infty}^{\infty}f(z)p(\frac{1}{2r}\log(1+2re^{-2rt}), h^{-1}(t,y),z)dz\nn \\
		&=&\int_{-\infty}^{\infty}f(z)q(e^{-2rt}, z-h^{-1}(t,y)\sqrt{1+2re^{-2rt}})dz. \label{e:h*fk}
		\eea
		Due to the bound on $f$ we can differentiate inside the integral sign to get
		\bea
		0&< &\frac{1}{\frac{d h^{-1}(s,y)}{dy}}\nn \\
		&=&\sqrt{1+2re^{-2rt}}\int_{-\infty}^{\infty}f(z)\frac{z-h^{-1}(t,y)\sqrt{1+2re^{-2rt}}}{e^{-2rt}} q(e^{-2rt}, z-h^{-1}(t,y)\sqrt{1+2re^{-2rt}})dz \nn \\
		&=&\sqrt{1+2re^{-2rt}}\int_{-\infty}^{\infty}f(z+h^{-1}(t,y)\sqrt{1+2re^{-2rt}})\frac{z}{e^{-2rt}} q(e^{-2rt}, z)dz \nn \\
		&\leq& C \exp\left(\frac{k}{2}(h^{-1}(t,y))^2(1+2r)\right)e^{3rt}\int_0^{\infty}\frac{z}{\sqrt{2\pi}}\exp\left(-\frac{z^2}{2}(e^{2rt}-k)\right)dz \nn \\
		&=&C \exp\left(\frac{k}{2}(h^{-1}(t,y))^2(1+2r)\right)\frac{e^{3rt}}{e^{2rt}-k}\nn\\
			&\sim&C \exp\left(\frac{k}{2}(h^{-1}(t,y))^2(1+2r)\right)e^{rt} \mbox{ as } t\rar \infty, \label{e:hybound}
		\eea
		where $C$  in above (and also throughout the proof) is a constant, independent of $t$, that might change from line to line. Moreover, $h^{-1}(t,y)= g^{-1}(V(t), y)$ is bounded for fixed $y$ since $V(\infty)<\infty$ and $g(,\cdot, y)$ is strictly increasing and continuous on $[0,V(\infty)]$ for each $y$. Thus, the above asymptotics  establishes that the condition is verified since 
		\[
		\int_0^{\infty}\sigma^*(s)<\infty.
		\]
		Thus, for any admissible $\alpha$ we have
		\[
		e^{-rt}J(t,Y_t)=J(0,0) + \int_0^t e^{-rs} (h(s,Y_s) -\Gamma)\left\{dB_s +\alpha_s ds\right\}
		\]
		due to the fact that $J_t +AJ -rJ=0$ by Lemma \ref{l:gJ}.  In view of the admissibility condition (\ref{adm:alpha}) 
		\[
		\left(\int_0^t e^{-rs} (h(s,Y_s) -\Gamma)dB_s\right)_{t \geq 0}
		\]
		is a uniformly integrable martingale martingale that converges in $L^1(d\bbP^v)$. Thus,  if $e^{-rt}J(t,Y_t)$ has a limit in $L^1(d\bbP^v)$ as $ t \rar \infty$, then
		\[
		\bbE^v \int_0^{\infty}  e^{-rs} (h(s,Y_s) -\Gamma)\alpha_sds=J(0,0)-\bbE^v \lim_{t \rar \infty}e^{-rt}J(t,Y_t). 
		\]
		Since $J\geq 0$, $\alpha$ will be the optimal strategy if it achieves $\lim_{t \rar \infty}e^{-rt}J(t,Y_t)=0$. To see that the above limit holds in $L^1(d\bbP^v)$,   first note that 
		\bea
		e^{-rt} \frac{1}{\sigma^*(t)}\int_{h^{-1}(t,v)}^y(h(t,x)-v)dx&\leq& \frac{\sqrt{1+2r}}{\sqrt{2r}}\int_{h^{-1}(t,v)}^y(h(t,x)-v)dx \nn \\
		&\leq& \frac{\sqrt{1+2r}}{\sqrt{2r}}(h(t,y)-v)(y-h^{-1}(t,v)) \label{e:L1bound}
		\eea
		due to (\ref{e:sigma}) and the fact that $h$ is increasing in $y$. Thus, $\lim_{t \rar \infty} e^{-rt}J(t,Y_t)=0,\, \bbP^v$-a.s. if $Y_t - h^{-1}(t,v)\rar 0,\, \bbP^v$-a.s.. However, Theorem \ref{t:gbridge} shows that $\alpha^*$ makes  $\lim_{t \rar \infty}Y^*_t=f^{-1}(v)=\lim_{t \rar \infty}h^{-1}(t,v), \, \bbP^v$-a.s.. Thus, if $\sup_{t \geq 0}\bbE^v |h^*(s,Y^*_s)|^{2+\eps}+ \bbE^v |Y^*_t|^{2+\epsilon}<\infty$, for some $\epsilon>0$, we can conclude that $e^{-rt}J(t,Y_t)$ converges to $0$ in   $L^1(d\bbP^v)$ as $t \rar \infty$ in  view of (\ref{e:L1bound}). Note that $\sup_{t \geq 0}\bbE^v (h^*(s,Y^*_s))^2<\infty$ will also imply that $\alpha^*$ is admissible. 
		
		We shall first show that $\bbE^v (Y^*_t)^{2+\eps}$ is bounded. Indeed, in view of the absolute continuity relationship established in Theorem \ref{t:gbridge}, we have
		\bean
		\bbE^v|Y^*_t|^{2+\eps} &=&\bbE^{\bbQ}\left[\frac{|Y^*_t|^{2+\eps}p(V(\infty)-V(t), Y_t^*,f^{-1}(v))}{p(V(\infty),0,f^{-1}(v))}\right]\\
		&=&\frac{1}{\sqrt{2\pi}p(V(\infty),0,f^{-1}(v))}\int_{-\infty}^{\infty}e^{rt}|y|^{2+\eps}p(V(\infty)-V(t),y,f^{-1}(v)) \exp\left(-\frac{y^2}{2  e^{-2rt}}\right)dy \\
		&\leq&C e^{-rt}\int_{-\infty}^{\infty}|y|^{\eps}p(V(\infty)-V(t),y,f^{-1}(v)) dy\\
		&\leq& C.
		 \eean
		   In above, the third line follows from the boundedness of $xe^{-x}$ on $(0,\infty)$ and the last line is due to the finiteness of $V(\infty)$.
		 
		To show $\sup_{t \geq 0}\bbE^v |h^*(s,Y^*_s)|^{2+\eps}<\infty$ first note that 
		\bea
		|h(t,y)|&\leq& C\exp\left(\frac{k}{2}y^2(1+2r)\right)\int_{-\infty}^{\infty} \frac{e^{rt}}{\sqrt{2\pi}}\exp\left(-\frac{z^2}{2}(e^{2rt}-k)\right)dz \nn\\
		&\leq&C\exp\left(\frac{k}{2}y^2(1+2r)\right)\frac{e^{rt}}{\sqrt{e^{2rt}-k}}\leq C\exp\left(\frac{k}{2}y^2(1+2r)\right),\label{e:hbound}
		\eea
		in view of the exponential bound on $f$ and  similar arguments that led to (\ref{e:hybound}).
		
		Therefore, setting $k^{\eps}=k (1+ \epsilon/2)$,
		 \bean
		\bbE^v |h^*(s,Y^*_s)|^{2+\eps}&\leq& C \int_{-\infty}^{\infty}\exp\left(k^{\eps}z^2(1+2r)\right)q\left(e^{-2rs},f^{-1}(v)-z\sqrt{1+2re^{-2rs}}\right)q\left(\frac{1}{1+2re^{-2rs}},z\right)dz\\
		&\leq & C \int_{-\infty}^{\infty}\exp\left(z^2\left\{(1+2r)k^{\eps}-\frac{1+2re^{-2rs}}{2}\right\}\right)q\left(e^{-2rs},f^{-1}(v)-z\sqrt{1+2re^{-2rs}}\right)dz\\
		&\leq&C \int_{-\infty}^{\infty}\exp\left(z^2e^{-2rs}\left\{\frac{(1+2r)k^{\eps}}{1+2re^{-2rs}}-\frac{1}{2}\right\}\right)\exp\left(-\frac{(f^{-1}(v)e^{rs}-z)^2}{2}\right)dz\\
		&\leq& C \exp\left((f^{-1}(v))^2\frac{m(s)}{{1-2e^{-2rs}m(s)}}\right),
		\eean
		where
		\[
		m(s):=\frac{k^{\eps}(1+2r)}{1+2re^{-2rs}}-\half<\half,
		\]
		if $\epsilon>0$ is chosen small enough. This completes the proof of that 	$\bbE^v |h^*(s,Y^*_s)|^{2+\eps}$ is bounded. Hence, $\alpha^*$ is an admissible optimal strategy.
		
		\item[Step 2.] {\em Market makers' best response.} Recall that
		\bean
		h^*(t,y)&=&\int_{-\infty}^{\infty}f(z)q(e^{-2rt}, z-y\sqrt{1+2re^{-2rt}})dz\\
		&=&\int_{-\infty}^{\infty}f(z+y\sqrt{1+2re^{-2rt}})q(e^{-2rt}, z)dz,
		\eean
		which shows that $h^*$ is strictly increasing since $f$ is. All that remains to show now is that $h^*(t,Y^*_t)$ is a $(\cF^{Y^*}_t, \bbP)$-martingale converging to $f(v)$.  
		
		It follows from Theorem \ref{t:gbridge} and the disintegration formula (\ref{e:disintegration}) that for any bounded and measurable $F$ we have
		\bean
		\bbE\left[F(Y^*_t)|\cF^{Y^*}_s\right]&=&\int_{\bbR}\bbE^v\left[F(Y^*_t)|\cF^{Y^*}_s\right]\nu(dv)\\
		&=&\int_{\bbR}\frac{\bbE^{\bbQ}\left[F(Y_t)p(V(\infty)-V(t),Y_t,f^{-1}(v))|\cF^{Y}_s\right]}{p(V(\infty)-V(t),Y_s,f^{-1}(v))}\nu(dv).
		\eean
		As we did before let $R_t=Y_{V^{-1}(t)}$, where $Y$ is the solution of (\ref{e:genYtc}). Thus,
		\[
		dR_t=d\beta_t +rR_tdt,
		\]
		for some Brownian motion $\beta$, which in particular implies $Y$ under $\bbQ$ is normal with mean $0$ and variance equalling
		\[
		\frac{e^{2rV(t)}-1}{2r}=\frac{1}{1+2re^{-2rt}}.
		\]
		 Thus,
		\bean
		\bbE\left[F(Y^*_t)|\cF^{Y^*}_s\right]&=&\int_{\bbR}\frac{\bbE^{\bbQ}\left[F(R_{V(t)})p(V(\infty)-V(t),R_{V(t)},f^{-1}(v))|\cF^{Y}_s\right]}{p(V(\infty)-V(t),R_{V(s)},f^{-1}(v))}\nu(dv)\\
		&=&\int_{\bbR}\int_{\bbR}f(y)\frac{p(V(\infty)-V(t),y,f^{-1}(v))p(V(t)-V(s),R_{V(s)},y)}{p(V(\infty)-V(s),R_{V(s)},f^{-1}(v))}dy\nu(dv)\\
		&=&\int_{\bbR}\bbE^{\bbQ}\left[F(R_{V(t)})\big|R_{V(s)}, R_{V(\infty)}=f^{-1}(v)\right]\nu(dv)\\
		&=&\bbE^{\bbQ}\left[F(R_{V(t)})\big|R_{V(s)}\right],
		\eean
		where the last line is due to the fact that $R_{V(\infty)}$ is a standard normal random variable due to the choice of $C^2$, and $f^{-1}(\Gamma)$ has a standard normal distribution by Assumption \ref{a:noatom}. Thus, the processes  $Y^*$ and $R_{V(\cdot)}$ have the same law. This  implies that $Y^*$ satisfies (\ref{e:Y*proj}) and, in particular, $h^*(t,Y^*_t)$ is a $(\cF^{Y^*}_t, \bbP)$-martingale. The convergence is immediate since $\lim_{t \rar \infty} Y^*_t=v$, $\bbP^v$-a.s., by Theorem \ref{t:gbridge}.
	\end{itemize}

\end{document}